\documentclass{article} 
\usepackage{hyperref}
\usepackage{amsthm}
\usepackage{amssymb}
\usepackage{xspace}
\usepackage{titletoc}
\usepackage[page, header]{appendix}

\usepackage{enumitem}
\usepackage{mathtools}
\usepackage{threeparttable}
\usepackage{booktabs}
\usepackage{multirow, multicol}

\usepackage{tikz}
\usetikzlibrary{positioning}
\usetikzlibrary{decorations.pathreplacing, decorations.pathmorphing}
\usetikzlibrary{arrows.meta}
\usetikzlibrary{patterns, patterns.meta}
\usetikzlibrary{calc}
\usetikzlibrary{shadows}
\usetikzlibrary{calligraphy}

\usepackage{authblk}

\usepackage[margin = 1.in]{geometry}

\title{Mixing Time Matters: Accelerating Effective Resistance Estimation via Bidirectional Method \\ {\Large \bfseries [Technical Report]}}

\author[1]{Guanyu Cui}
\author[2]{Hanzhi Wang}
\author[1]{Zhewei Wei\thanks{Zhewei Wei is the corresponding author.}}
\affil[1]{Renmin University of China, China}
\affil[2]{BARC, University of Copenhagen, Denmark}

\date{}

\usepackage{mathtools}
\usepackage{amsthm}
\usepackage{multirow}
\usepackage{multicol}
\usepackage{enumitem}
\setlist[itemize]{leftmargin=*}
\usepackage{pgf}
\usepackage{balance}

\usepackage{tikz}
\usetikzlibrary{positioning}
\usetikzlibrary{decorations.pathmorphing}

\usepackage[ruled, vlined, linesnumbered]{algorithm2e}

\SetCommentSty{mycommfont}

\theoremstyle{plain}
\newtheorem{theorem}{Theorem}[section]

\newtheorem{lemma}[theorem]{Lemma}

\theoremstyle{definition}
\newtheorem{definition}[theorem]{Definition}

\theoremstyle{remark}

\newcommand\bs[1]{\boldsymbol{#1}}

\renewcommand{\Pr}[1]{\mathrm{Pr}\left(#1\right)}
\newcommand{\E}[1]{\mathbf{E}\left[#1\right]}
\newcommand{\Var}[1]{\mathbf{Var}\left[#1\right]}
\newcommand{\Given}{\;\Big|\;}
\newcommand{\A}{\mathbf{A}}
\newcommand{\D}{\mathbf{D}}
\newcommand{\dmax}{d_{\max}}
\newcommand{\dmin}{d_{\min}}
\newcommand{\dbar}{\bar{d}}
\renewcommand{\P}{\mathbf{P}}
\renewcommand{\L}{\mathbf{L}}
\newcommand{\Lmax}{L_{\max}}
\newcommand{\q}{\bs{q}}
\renewcommand{\r}{\bs{r}}
\newcommand{\mathdefault}[1][]{}

\def\header{\noindent}

\begin{document}

\maketitle

\begin{abstract}
We study the problem of efficiently approximating the \textit{effective resistance} (ER) on undirected graphs, where ER is a widely used node proximity measure with applications in graph spectral sparsification, multi-class graph clustering, network robustness analysis, graph machine learning, and more. Specifically, given any nodes $s$ and $t$ in an undirected graph $G$, we aim to efficiently estimate the ER value $R(s,t)$ between nodes $s$ and $t$, ensuring a small absolute error $\epsilon$. The previous best algorithm for this problem has a worst-case computational complexity of $\tilde{O}\left(\frac{L_{\max}^3}{\epsilon^2 d^2}\right)$, where the value of $L_{\max}$ depends on the mixing time of random walks on $G$, $d = \min\{d(s), d(t)\}$, and $d(s)$, $d(t)$ denote the degrees of nodes $s$ and $t$, respectively. We improve this complexity to $\tilde{O}\left(\min\left\{\frac{L_{\max}^{7/3}}{\epsilon^{2/3}}, \frac{L_{\max}^3}{\epsilon^2d^2}, mL_{\max}\right\}\right)$, achieving a theoretical improvement of $\tilde{O}\left(\max\left\{\frac{L_{\max}^{2/3}}{\epsilon^{4/3} d^2}, 1, \frac{L_{\max}^2}{\epsilon^2 d^2 m}\right\}\right)$ over previous results. Here, $m$ denotes the number of edges. Given that $L_{\max}$ is often very large in real-world networks (e.g., $L_{\max} > 10^4$), our improvement on $L_{\max}$ is significant, especially for real-world networks. We also conduct extensive experiments on real-world and synthetic graph datasets to empirically demonstrate the superiority of our method. The experimental results show that our method achieves a $10\times$ to $1000\times$ speedup in running time while maintaining the same absolute error compared to baseline methods. 
\end{abstract}
\section{Introduction}
Effective Resistance (ER) is a widely adopted measure for node proximity in graphs, with applications across various scenarios. 
For example, ER approximation is closely related to optimal transport problems~\cite{robertson2024all}, the detection of low conductance sets in graph clustering~\cite{alev2018graph}, the maximum flow problem~\cite{christiano2011electrical}, and graph spectral sparsification~\cite{spielman2008graph}.
Additionally, the ER metric is utilized to enhance seeding strategies in influence maximization~\cite{hong2024new} and to assess network robustness against disruptions~\cite{yamashita2021effective}. 
Moreover, with the growing prominence of graph neural networks (GNNs), ER is increasingly employed in the graph rewiring process of GNNs to address the oversquashing problem~\cite{black2023understanding} and to improve the expressive power of GNN models~\cite{zhang2023rethinking}. 

Recognizing these widespread applications, the problem of efficiently estimating ER scores between a given pair of nodes has been the subject of extensive research~\cite{peng2021local, yang2023efficient, liao2023efficient}. 
Specifically, considering an undirected graph $G$ with two arbitrary nodes $s$ and $t$ in $G$, the ER score between $s$ and $t$, denoted by $R(s, t)$, is defined as
\begin{equation}
\label{eq:ER_def}
R(s, t) = (\bs{e}_s - \bs{e}_t)^\top\L^+(\bs{e}_s - \bs{e}_t) = \L^+_{ss} - \L^+_{st} - \L^+_{ts} + \L^+_{tt}, 
\end{equation}
where $\L^+$ denotes the Moore-Penrose pseudo-inverse of the graph $G$'s Laplacian matrix $\L$, and $\bs{e}_s$ denotes a one-hot vector with the $s$-th element being $1$. 

To understand this definition formula more intuitively, Peng et al.~\cite{peng2021local} provide a probabilistic interpretation of ER scores, showing that Equation~\eqref{eq:ER_def} can be rewritten as 
\begin{equation}
\label{eq:relation_transition_ER}
    R(s, t) = \displaystyle\sum\limits_{\ell = 0}^{\infty} \left(\dfrac{p^{(\ell)}(s, s)}{d(s)} - \dfrac{p^{(\ell)}(s, t)}{d(t)} - \dfrac{p^{(\ell)}(t, s)}{d(s)} + \dfrac{p^{(\ell)}(t, t)}{d(t)}\right), 
\end{equation}
where $p^{(\ell)}(s, t)$ denotes the transition probability of a random walk on $G$ moving from node $s$ to $t$ in its $\ell$-th step, and $d(s)$ corresponds to the degree of node $s$ in $G$. 
In other words, the ER score between nodes $s$ and $t$ equals the sum of bidirectional degree-normalized random walk probabilities. 
This interpretation has inspired a line of research~\cite{peng2021local, yang2023efficient} focused on estimating ER scores by leveraging techniques for computing random walk transition probabilities on graphs. 
This approach involves setting a maximum length $\Lmax$ for random walks, estimating the $\ell$-hop transition probabilities between nodes $s$ and $t$ for each $\ell \in [0, \Lmax]$ (i.e., estimating $p^{(\ell)}(s,t)$, $p^{(\ell)}(t,s)$, $p^{(\ell)}(s,s)$, and $p^{(\ell)}(t,t)$), then calculating the $\Lmax$-truncated ER, denoted as $R_{\Lmax}(s, t)$, as the estimate for $R(s,t)$. 
The definition of $R_{\Lmax}(s, t)$ is provided below:
\begin{equation}
\label{eq:L-truncated}
    R_{\Lmax}(s, t) = \displaystyle\sum\limits_{\ell = 0}^{\Lmax} \left(\dfrac{p^{(\ell)}(s, s)}{d(s)} - \dfrac{p^{(\ell)}(s, t)}{d(t)} - \dfrac{p^{(\ell)}(t, s)}{d(s)} + \dfrac{p^{(\ell)}(t, t)}{d(t)}\right). 
\end{equation}

Notably, due to the mixing time of random walks on graphs, the $\ell$-hop transition probability converges after several steps. 
Consequently, as shown in Equation~\eqref{eq:L-truncated}, this line of research~\cite{peng2021local, yang2023efficient} truncates the infinite summation in Equation~\eqref{eq:relation_transition_ER} at a level of $\Lmax$ and focuses on approximating the $\ell$-hop transition probabilities for $\ell \le \Lmax$. 
As a result, to estimate the ER score for a given pair of nodes $s$ and $t$ while assuring an additive error $\epsilon$ and a constant failure probability, the best-known methods~\cite{yang2023efficient}, AMC and GEER, both achieve a worst-case computational complexity of $\tilde{O}\left(\frac{\Lmax^3}{\epsilon^2 d^2}\right)$, where $d=\min\left\{d(s), d(t)\right\}$. 

{\header \bf Our Motivations.}
As proved in~\cite{peng2021local}, the value of $\Lmax$ depends on the additive error $\epsilon$ specified in the approximation problem and the spectral radius of the graph. 
Although previous studies~\cite{peng2021local, yang2023efficient} claim that the \textit{order} of $\Lmax$ for real-world graphs is $O(\log n)$, the actual \textit{value} of $\Lmax$ can still be very large. 
As reported in Table~\ref{tab:L-vs-eps}, the value of $\Lmax$ always exceeds $10^3$ or even $10^4$. 
This large value of $\Lmax$ makes the $O(\Lmax^3)$ time-dependence in the previously best complexity result~\cite{yang2023efficient} infeasible on real-world large graphs. 
To address this challenge, existing studies often set it to a much smaller value than its true value in experiments. 
For example, on the Facebook dataset with an additive error $\epsilon = \texttt{1e-3}$, the actual value of $\Lmax$ is $14546$ as given in Table~\ref{tab:L-vs-eps}, while previous works set $\Lmax$ to $74$ in ~\cite{peng2021local} and $96$ in ~\cite{yang2023efficient} in experiments~\footnote{
The value of $\Lmax$ depends on the spectral radius $\lambda$ of $G$ and the additive error $\epsilon$ specified in the approximation problem. For the Facebook dataset, the true spectral radius is $\lambda = 0.9992$, which results in $\Lmax = 14,546$ with $\epsilon = 10^{-3}$. 
In contrast, previous work~\cite{peng2021local} and~\cite{yang2023efficient} used $\lambda = 0.9$ and $\lambda = 0.9192$, respectively, leading to much smaller values of $\Lmax$, specifically $74$ and $96$.}. 
However, it is important to note that the values of $R_{\Lmax}(s,t)$ with different $\Lmax$ values vary significantly. 
Table \ref{tab:truncated-vs-Lmax} lists the value of $R_{\Lmax}(s,t)$ for different $\Lmax$. 
We observe that $R_{\Lmax}(s,t)$ only converges to within an additive error of $\texttt{1e-3}$ after $\Lmax\ge 2,000$. 
In other words, to estimate the ER score $R(s,t)$ with an additive error of $\texttt{1e-3}$, $\Lmax$ must be at least $2,000$. 
Using smaller values of $\Lmax$ in experiments, as done in previous methods~\cite{peng2021local, yang2023efficient}, can lead to unsatisfactory approximation accuracy. 
This highlights the urgent need to improve the time dependence on $\Lmax$ for ER approximation.
\begin{table}[t]
    \caption{$\Lmax$ across various real-world datasets, assuming $\epsilon = \texttt{1e-3}$ and $d(s) = d(t) = \left\lceil\bar{d}\right\rceil$.}
    \label{tab:L-vs-eps}
    \centering
    \begin{tabular}{cccccc}
        \toprule
        Facebook & DBLP & Youtube & Orkut & LiveJournal & Friendster \\
        \midrule
        14,546 & 4,536 & 6,353 & 1,767 & 146,133 & 23,734 \\
        \bottomrule
    \end{tabular}
\end{table}

\begin{table}[t]
    \centering
    \caption{$\Lmax$-truncated ER values (with 4 significant figures) versus $\Lmax$ on the Facebook dataset.}
    \label{tab:truncated-vs-Lmax}
    \begin{tabular}{ccccccc}
        \toprule
        $\Lmax$ & 100 & 1,000 & 2,000 & 3,000 & 4,000 & 5,000 \\
        \midrule
        {$R_{\Lmax}(s,t)$}& $0.1517$ & $0.1791$ & $0.1809$ & $0.1814$ & $0.1815$ & $0.1816$ \\
        \bottomrule
    \end{tabular}
\end{table}

\subsection{Our Contributions}
To address the aforementioned challenges, we present several contributions in this paper. 

First, we consider the problem of estimating the ER score $R(s,t)$ for an arbitrary pair of nodes $s$ and $t$ in an undirected graph $G$. 
We improve the worst-case computational complexity from the previous best result of $\tilde{O}\left(\frac{\Lmax^3}{\epsilon^2 d^2}\right)$ to $\tilde{O}\left(\min\left\{\frac{\Lmax^{7/3}}{\epsilon^{2/3}}, \frac{\Lmax^3}{\epsilon^2 d^2}, m\Lmax\right\}\right)$, where $d = \min\{d(s), d(t)\}$ and $m$ denotes the number of edges in $G$. 
Our result is asymptotically better than previous results, achieving a factor of $\tilde{O}\left(\max\left\{\frac{\Lmax^{2/3}}{\epsilon^{4/3} d^2}, 1, \frac{\Lmax^2}{\epsilon^2 d^2 m}\right\}\right)$ improvements. 
Notably, in real-world networks where node degrees often follow power-law distributions—with most nodes having small degrees—our theoretical improvements are further enhanced due to the typically small values of $d = \min\{d(s), d(t)\}$. 

Moreover, we conduct extensive experiments on real-world and synthetic large graphs to show the superiority of our algorithm. 
Experimental results show that our algorithm outperforms all baseline methods by up to an order of magnitude in both efficiency and accuracy. 

In particular, we hope to draw attention to the time dependence on $\Lmax$ in ER approximation. 
We emphasize that $\Lmax$ is influenced by the mixing time of random walks and the spectral radius of the graph, and the true value of $\Lmax$ for real-world networks can be very large. 
Therefore, improving the dependence on $\Lmax$ has substantial theoretical and practical implications.
\section{Preliminaries}
\label{sec:preliminaries}

\header{\bf Notations.}
We consider an undirected connected graph $G = (V, E)$ with $n$ nodes and $m$ edges. 
The neighborhood of a node $i$, denoted as $N(i)$, consists of nodes that share an edge with it, i.e., $N(i) = \{j: \{i, j\}\in E\}$.
The degree of a node $i$, denoted as $d(i)$, is the number of nodes in its neighborhood.
We use $\dmax$ to represent the maximum degree, $\dmin$ for the minimum degree, and $\dbar$ for the average degree.
The adjacency matrix $\A$ is defined as $\A_{ij} = \mathbf{I}[(i, j) \in E]$.
The degree matrix $\D$ is a diagonal matrix with $\D_{ii} = d(i)$.
The Laplacian matrix $\L$ is defined as $\D - \A$.
Additionally, the Moore-Penrose pseudo-inverse of $\L$ is denoted as $\L^+$.
Moreover, $\bs{e}_s$ denotes the one-hot vector with the $s$-th element being $1$. 
Table \ref{tab:notations} summarizes frequently used notations for reference.
\begin{table}[t]
    \caption{Frequently used notations.}
    \label{tab:notations}
    \centering
    \begin{tabular}{l|l}
        \toprule
        \textbf{Notation} & \textbf{Description} \\
        \midrule
        $G = (V, E)$ & undirected graph with node set $V$ and edge set $E$. \\
        $n$, $m$ & number of nodes and edges. \\
        $N(i)$ & neighborhood of node $i$. \\
        $d(i)$ & degree of node $i$. \\
        $\dmax$, $\dmin$, $\dbar$ & maximum, minimum, and average degree. \\
        $\A$ & adjacency matrix of $G$. \\
        $\D$ & degree matrix of $G$. \\
        $\L = \D - \A$ & Laplacian matrix of $G$. \\
        $\L^+$ & Moore-Penrose inverse of $\L$. \\
        $\bs{e}_s$ & one-hot vector whose $s$-th element is $1$. \\
        $\P = \A\D^{-1}$ & transition matrix of $G$. \\
        $\lambda_i$ & $i$-th largest eigenvalue of $\P$. \\
        $p^{(\ell)}(s, t)$ & $\ell$-hop transition probability from $s$ to $t$. \\
        $h(s, t)$ & hitting time from $s$ to $t$. \\
        $\kappa(s, t)$ & commute time between $s$ and $t$. \\
        $R(s, t)$ & effective resistance between $s$ and $t$. \\
        $\Lmax$ & maximum random walk length (see Lemma \ref{lem:L}). \\
        $\epsilon$, $p_f$ & absolute error and the failure probability. \\
        $O(\cdot)$ & big-Oh notation, asymptotic upper bound. \\ 
        $\tilde{O}(\cdot)$ & soft big-Oh notation, ignoring some log factors. \\
        \bottomrule
    \end{tabular}
\end{table}

\header{\bf Problem Definition. } 
In this paper, we consider the problem of estimating Single-Pair Effective Resistance (SPER), as defined below. 
\begin{definition}[SPER Estimation with Absolute Error Guarantee]
Given an undirected connected graph $G = (V, E)$, an arbitrary pair of nodes $s, t\in V$, an absolute error tolerance $\epsilon > 0$, and a failure probability $0 < p_f\le 1$, the objective of a SPER query with an absolute error guarantee is to provide an estimator $\hat{R}(s, t)$ such that $$\Pr{\left|\hat{R}(s, t) - R(s, t)\right| < \epsilon} \ge 1 - p_f.$$
\end{definition}

\subsection{Probabilistic Interpretation}
The ER score can be interpreted as transition probabilities of random walks on graphs. 
Specifically, the random walk on graphs is an important stochastic process. 
It operates as follows: starting from a node $v_0$, at each time step, when at node $v_t$, the next step involves moving to any neighboring node of $v_t$ with a probability of ${1}/{d(v_t)}$.
This random walk can also be viewed as a Markov chain with transition matrix $\P = \A\D^{-1}$.
Let $1 = \lambda_1 \ge \lambda_2 \ge \cdots \ge \lambda_{n} \ge - 1$ denote the eigenvalues of $\P$ sorted in descending order. 
The $\ell$-hop transition probability $p^{(\ell)}(s, t)$ represents the probability that a random walk starting at node $s$ visits node $t$ after $\ell$ hops. 
It's noteworthy that the $\ell$-hop transition probability matrix $\P^{(\ell)} = (p^{(\ell)}(i, j))_{ij}$ is equivalent to the $\ell$-th power of $\P$, i.e., $\P^{(\ell)} = \P^{\ell}$.
We will also highlight a well-known symmetric property concerning the multi-step transition probabilities (MSTP) on undirected graphs as follows:
\begin{lemma}[Symmetric Property of MSTP, \cite{lofgren2015bidirectional}]
\label{lem:symmetry}
    Given any undirected graph $G = (V, E)$, for any node $s$ and $t$, and for all $\ell$, the following property of the $\ell$-hop transition probabilities holds:
    $$\dfrac{p^{(\ell)}(s, t)}{d(t)} = \dfrac{p^{(\ell)}(t, s)}{d(s)}.$$
\end{lemma}
The hitting time from $s$ to $t$, denoted as $h(s, t)$, indicates the expected steps for a random walk starting from node $s$ to hit node $t$ for the first time. 
The commute time between $s$ and $t$, denoted as $\kappa(s, t)$, is defined as $\kappa(s, t) = h(s, t) + h(t, s)$.

In particular, Peng et al.~\cite{peng2021local} provide Equation~\eqref{eq:relation_transition_ER} to illustrate the connection between the Moore-Penrose pseudo-inverse of the Laplacian matrix $\L^+$ and the MSTP matrix $\P^{\ell}$, given that $(\bs{e}_s - \bs{e}_t)^\top\L^+(\bs{e}_s - \bs{e}_t) = \sum\limits_{\ell = 0}^{\infty}(\bs{e}_s - \bs{e}_t)^\top\D^{-1}\P^{\ell}(\bs{e}_s - \bs{e}_t)$.
A line of research~\cite{peng2021local, yang2023efficient}
use the $\Lmax$-truncated ER $R_{\Lmax}(s, t)$ as defined in Equation~\eqref{eq:L-truncated} to compute the estimate for $R(s,t)$.

Additionally, we need to bound the truncation error to satisfy the absolute error constraint.
Peng et al. have proved such a bound in \cite{peng2021local}, which is later refined by Yang et al. in \cite{yang2023efficient}.
The bound from \cite{yang2023efficient} is listed below:
\begin{lemma}[Maximum Steps Needed, \cite{yang2023efficient}]
\label{lem:L}
    Given an undirected graph $G$, $\left|R(s, t) - R_L(s, t)\right| \le \frac{\epsilon}{2}$ holds for $s$ and $t$ when $L$ satisfies
    \begin{equation}
    \label{eq:L}
        L \ge \Lmax = \left\lceil\log_{{1}/{\lambda}}\dfrac{2\left({1}/{d(s)} + {1}/{d(t)}\right)}{\epsilon(1 - \lambda)}\right\rceil,
    \end{equation}
    where $\lambda = \max\{\lambda_2, |\lambda_n|\}$.
\end{lemma}
Notably, there is a strong connection between $\Lmax$ and the mixing time $\tau_{\text{mix}}(\epsilon)$, a crucial property of graphs.
The mixing time of an undirected graph quantifies how quickly a random walk on the graph approaches the stationary distribution $\bs{\pi} = \frac{\bs{D1}}{2m}$. 
It can be formally defined as 
$$\tau_{\text{mix}}(\epsilon) = \min\left\{t: \max\limits_{x \in V}\left\|\bs{P}^t\bs{e}_x - \bs{\pi}\right\|_{\text{TV}} \le \epsilon \right\},$$ 
where $\left\|\bs{P}^t\bs{e}_x - \bs{\pi}\right\|_{\text{TV}} = \frac{1}{2}\sum\limits_{y \in V} \left|p^{(t)}(x, y) - \bs{\pi}(y)\right|$ is the total variation distance between two distributions.
It can be shown that the mixing time is bounded by $O\left(\frac{\log\left(\frac{m}{\dmin\epsilon}\right)}{1 - \lambda}\right)$.
Additionally, according to Equation \eqref{eq:L}, we have $\Lmax \le \left\lceil\log_{\frac{1}{\lambda}}\frac{4}{\dmin\epsilon(1 - \lambda)}\right\rceil$. 
Since for any $c > 0$, we also have $\lim\limits_{\lambda \to 1}{\frac{\log\left(A / (1 - \lambda)\right)}{\log(1 / \lambda)}} \Big/ \frac{A}{(1 - \lambda)^{1 + c}} = 0$, we can bound $\Lmax$ as $O\left(\frac{\log\left(\frac{1}{\dmin\epsilon}\right)}{(1 - \lambda)^{1 + c}}\right)$, which shows a strong connection between $\Lmax$ and the mixing time $\tau_{\text{mix}}(\epsilon)$.
In fact, in the proof of Lemma \ref{lem:L}~\cite{yang2023efficient}, the authors merely change the definition of ``close'' by replacing the total variation distance with the $L_{\infty}$ norm.

\subsection{Key Techniques}
According to the previous subsection, ER has a strong connection to MSTPs.
In this section, we introduce some key techniques for estimating MSTPs, including the forward push operation, Monte Carlo sampling, and bidirectional methods.
Notably, our method also utilizes some of these techniques.

{\header \bf Forward Push Operation.}
The forward push operation is a deterministic procedure to spread probability masses on graphs, originally proposed by Anderson et al. \cite{andersen2006local} for PageRank vector computation.
Here, we introduce it in the context of MSTP computation.

The core idea of the forward push operation, with respect to a node $s$, involves maintaining two types of vectors for each layer (number of steps) $0 \leq \ell \leq \Lmax$: reserve vectors $\q^{(\ell)}_s \in \mathbb{R}^{n}$ and residue vectors $\r^{(\ell)}_s \in \mathbb{R}^{n}$. 
The $u$-th element of a reserve vector $\q^{(\ell)}_s$ represents the accumulated probability mass on node $u$ at layer $\ell$ and serves as an underestimate of $p^{(\ell)}(s, u)$.
The $u$-th element of a residue vector $\r^{(\ell)}_s$ indicates the current active probability mass on node $u$ at layer $\ell$ that will be distributed to neighboring nodes in the next layer.

\begin{algorithm}[t]
\DontPrintSemicolon
    \caption{Forward-Push$_s(u, \ell)$}
    \label{alg:Forward-Push-MSTP}
    $\q_s^{(\ell)}(u) \gets \q_s^{(\ell)}(u) + \r_s^{(\ell)}(u)$ \\
    \For{$v\in N(u)$}
    {
        $\r_s^{(\ell + 1)}(v) \gets \r_s^{(\ell + 1)}(v) + \frac{\r_s^{(\ell)}(u)}{d(u)}$
    }
    $\r_s^{(\ell)}(u) \gets 0$
\end{algorithm}
The Forward-Push procedure, detailed in Algorithm \ref{alg:Forward-Push-MSTP}, operates as follows: invoking Forward-Push$_s(u, \ell)$ increases $\bs{q}_s^{(\ell)}(u)$ by $\bs{r}_s^{(\ell)}(u)$, then distributes $\bs{r}_s^{(\ell)}(u)$ evenly to its neighbors by increasing $\bs{r}_s^{(\ell + 1)}(v)$ by $\frac{\bs{r}_s^{(\ell)}(u)}{d(u)}$ for each $v \in N(u)$, and finally set $\bs{r}_s^{(\ell)}(u)$ to $0$.
We also present a running example in Figure \ref{fig:Forward-Push-Example} where node $u$ has three neighbors.

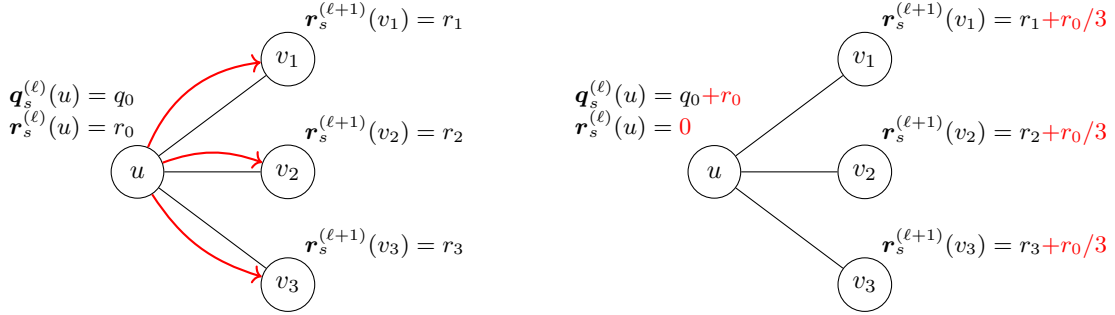
\begin{figure}[t]
    \caption{A running example of Forward-Push$_s(u, \ell)$.}
    \label{fig:Forward-Push-Example}
    \centering
    \begin{tikzpicture}[scale = 1, transform shape, every node/.style={minimum size = .7cm}]
        \node[circle, draw] (u) at (0, 0) {$u$};
        \node[anchor = south east, align = left] at ([shift = {(.1cm, .3cm)}]u){\small $\bs{q}_s^{(\ell)}(u) = q_0$ \\ \small $\bs{r}_s^{(\ell)}(u) = r_0$};
        \node[circle, draw] (v1) at (2., 1.5) {$v_1$};
        \node[anchor = south west] at ([shift = {(.1cm, .2cm)}]v1) {\small $\bs{r}_s^{(\ell + 1)}(v_1) = r_1$};
        \node[circle, draw] (v2) at (2., 0) {$v_2$};
        \node[anchor = south west] at ([shift = {(.1cm, .2cm)}]v2){\small $\bs{r}_s^{(\ell + 1)}(v_2) = r_2$};
        \node[circle, draw] (v3) at (2., -1.5) {$v_3$};
        \node[anchor = south west] at ([shift = {(.1cm, .2cm)}]v3){\small $\bs{r}_s^{(\ell + 1)}(v_3) = r_3$};
        \draw (u) -- (v1);
        \draw[red, ->, thick, bend left = 30] (u) to (v1);
        \draw (u) -- (v2);
        \draw[red, ->, thick, bend left = 20] (u) to (v2);
        \draw (u) -- (v3);
        \draw[red, ->, thick, bend right = 20] (u) to (v3);
    \end{tikzpicture}
    \hspace{1.0cm}
    \begin{tikzpicture}[scale = 1, transform shape, every node/.style={minimum size = .7cm}]
        \node[circle, draw] (u) at (0, 0) {$u$};
        \node[anchor = south east, align = left] at ([shift = {(.5cm, .3cm)}]u){\small $\bs{q}_s^{(\ell)}(u) = q_0 \color{red}{+ r_0}$ \\ \small $\bs{r}_s^{(\ell)}(u) = \color{red}{0}$};
        \node[circle, draw] (v1) at (2., 1.5) {$v_1$};
        \node[anchor = south west] at ([shift = {(.1cm, .2cm)}]v1){\small $\bs{r}_s^{(\ell + 1)}(v_1) = r_1\color{red}{+r_0/3}$};
        \node[circle, draw] (v2) at (2., 0) {$v_2$};
        \node[anchor = south west] at ([shift = {(.1cm, .2cm)}]v2){\small $\bs{r}_s^{(\ell + 1)}(v_2) = r_2\color{red}{+r_0/3}$};
        \node[circle, draw] (v3) at (2., -1.5) {$v_3$};
        \node[anchor = south west] at ([shift = {(.1cm, .2cm)}]v3){\small $\bs{r}_s^{(\ell + 1)}(v_3) = r_3\color{red}{+r_0/3}$};
        \draw (u) -- (v1);
        \draw (u) -- (v2);
        \draw (u) -- (v3);
    \end{tikzpicture}
\end{figure}

{\header \bf Monte Carlo Sampling.}
Monte Carlo sampling is a classic technique for estimating probabilistic quantities.
To estimate a quantity $r$, we design a random variable $X$ such that $\E{X} = r$ and $\Var{X} < \infty$.
By sampling $N$ independent instances $(X_i)_{i = 1}^N$, the empirical mean $\frac{1}{N}\sum\limits_{i = 1}^N X_i$ serves as an estimate for $r$.
In the context of SPER estimation, Monte Carlo sampling involves generating multiple $\Lmax$ random walks to estimate MSTPs, which then provides an estimator for $R(s, t)$.

{\header \bf Bidirectional Methods.}
Banerjee et al.~\cite{banerjee2015fast} introduce a bidirectional algorithm named Bidirectional-MSTP (BiMSTP for short) for approximating MSTPs; and similar techniques have also been applied to (Personalized) PageRank approximation~\cite{lofgren2015bidirectional, lofgren2016personalized, wang2024revisiting, yang2024efficient}. 
Banerjee et al.~\cite{banerjee2015fast} prove an invariant (Lemma \ref{lem:invariant}), and integrate the forward push operation with Monte Carlo sampling to estimate the $\ell$-hop transition probability $p^{(\ell)}(s, t)$.
\begin{lemma}[Invariant, \cite{banerjee2015fast}]
\label{lem:invariant}
If we initialize the reserve and residue vectors as $\q_s^{(\ell)} = \bs{0}$ for $\ell = 0, 1, \cdots, \Lmax$, and $\r_s^{(0)} = \bs{e}_s$, $\r_s^{(\ell)} = \bs{0}$ for $\ell = 1, 2, \cdots, \Lmax$, then after any sequence of Forward-Push$_s(u, \ell)$ operations, for any $0 \le \ell \le \Lmax$, the following invariant holds:
\begin{align*}
    p^{(\ell)}(s, t) = \q_s^{(\ell)}(t) + \sum_{k = 0}^{\ell}\sum_{v\in V}\r_s^{(\ell - k)}(v)p^{(k)}(v, t).
\end{align*}
\end{lemma}
\section{Related Works}
\label{sec:related-works}
Existing works that focus on SPER estimation can be categorized into four types: estimating multi-step transition probabilities, utilizing landmark nodes, estimating commute time, and solving the Laplacian system of equations.
In this section, we begin by introducing transition-probabilities-based methods, which are fast, have an absolute error guarantee, and are easy to implement. 
Following this, we provide a concise overview of other relevant lines of research.

\subsection{Transition-Probabilities-Based Methods}
In their work~\cite{peng2021local}, Peng et al. introduce the EstEff-TranProb algorithm, which employs random walk simulations to approximate the transition probabilities, subsequently deriving an estimator for $R(s, t)$ with a computational complexity of $\tilde{O}\left(\frac{\Lmax^4}{\epsilon^2}\right)$. 
Building on this, Yang et al. \cite{yang2023efficient} suggest an approach that adaptively samples random walks, applying Bernstein's inequality for early termination. 
This innovation led to the development of the AMC algorithm and an enhanced variant, GEER. 
Both algorithms achieve a time complexity of $\tilde{O}\left(\frac{\Lmax^3}{\epsilon^2 d^2}\right)$, where $d = \min\{d(s), d(t)\}$.

\subsection{Landmark-Based Methods}
In their recent study, Liao et al.~\cite{liao2023efficient} introduce a novel technique for computing the SPER, proposing four distinct algorithms that leverage the concept of a landmark node. 
Subsequently, Liao et al.~\cite{liao2024efficient} expand the algorithms by moving from the use of a single landmark node to incorporating a set of landmark nodes.
However, for all proposed algorithms except one, it is not possible to directly set the algorithm parameters to get an estimator with an absolute error guarantee, making these methods heuristic and lacking guaranteed error approximation.
For the exception LocalTree algorithm, it has a worst-case running time of $\tilde{O}\left(\frac{n^2(h(s, v) + h(t, v))}{\epsilon^2}\right)$, where $v$ is the landmark node. 
Additionally, the efficiency of all four algorithms is dependent on the choice of the landmark node, a decision that largely relies on heuristic approaches. 

\subsection{Commute-Time-Based Methods}
Peng et al. \cite{peng2021local} capitalized on the relation between the ER and the commute time, as detailed in \cite{lovasz1993random}, to develop their EstEff-MC algorithm. 
This algorithm provides an estimation of $R(s, t)$ achieving a relative error margin within $\epsilon$. 
The algorithm's worst-case expected running time is formulated as $\tilde{O}\left(\frac{m}{(1 - \lambda_2)^2d \epsilon^2}\right)$, showcasing a direct proportionality to $m$ and a quadratic dependence on $\frac{1}{\epsilon}$.

\subsection{Laplacian-Solver-Based Methods}
Per Equation \eqref{eq:ER_def}, computing $R(s, t)$ involves solving the linear system $\L \bs{x} = \bs{e}_s - \bs{e}_t$ to obtain $\L^+(\bs{e}_s - \bs{e}_t)$, and then subtracting the $t$-th element from its $s$-th element.
The advancements in \cite{spielman2004nearly, koutis2010approaching, koutis2011nearly, cohen2014solving, jambulapati2021ultrasparse} have achieved a nearly-linear time complexity. 
This enables the computation of an estimator for $R(s, t)$ within $\tilde{O}\left(m\log \frac{1}{\epsilon}\right)$ time. 

\begin{table}[t]
    \caption{Time complexity of the algorithms.}
    \label{tab:complexity}
    \centering
    \begin{tabular}{cc}
        \toprule
        \textbf{Method} & \textbf{Query Time} \\
        \midrule
        EstEff-TranProb~\cite{peng2021local} & $\tilde{O}\left(\frac{\Lmax^4}{\epsilon^2}\right)$ \\
        AMC / GEER~\cite{yang2023efficient} & $\tilde{O}\left(\frac{\Lmax^3}{\epsilon^2 d^2}\right)$ \\
        EstEff-MC~\cite{peng2021local} & $\tilde{O}\left(\frac{m}{(1 - \lambda_2)^2 \epsilon^2 d}\right)$ \\
        Lap. Solvers~\cite{spielman2004nearly, koutis2010approaching, koutis2011nearly, cohen2014solving, jambulapati2021ultrasparse} & $\tilde{O}\left(m\right)$ \\
        \midrule
        BiSPER \textbf{(Ours)} & $\tilde{O}\left(\min\left\{\frac{\Lmax^{7/3}}{\epsilon^{2/3}}, \frac{\Lmax^3}{\epsilon^2d^2}, m\Lmax\right\}\right)$ \\
        \bottomrule
    \end{tabular}
\end{table}

\section{Our Approach}
\label{sec:our-approach}
In this section, we present our methodology for estimating the ($\Lmax$-truncated) SPER $R_{\Lmax}(s, t)$. 
We introduce the BiSPER algorithm, a novel approach that combines a redesigned forward push operation with adaptive Monte Carlo sampling techniques. 

\subsection{High Level Idea}
By integrating Equation~\eqref{eq:L-truncated}, Lemma \ref{lem:invariant}, and the symmetry property of MSTP in Lemma \ref{lem:symmetry}, we can derive a direct estimator for $R_{\Lmax}(s, t)$, using a number ($N$) of random walks:
\begin{align}
\label{eq:L-truncated-estimator}
    \begin{split}
    \hat{R}_{\Lmax}(s, t) 
     = \sum_{\ell = 0}^{\Lmax}\left(\dfrac{\bs{q}_s^{(\ell)}(s)}{d(s)} - \dfrac{\bs{q}_s^{(\ell)}(t)}{d(t)}\right) + \sum_{\ell = 0}^{\Lmax}\left(\dfrac{\bs{q}_t^{(\ell)}(t)}{d(t)} - \dfrac{\bs{q}_t^{(\ell)}(s)}{d(s)}\right) & \\
    + \dfrac{1}{N}\sum_{i = 1}^{N}\sum_{\ell = 0}^{\Lmax}\sum_{v\in V}X^{(\ell)}_{s, i}(v)\left(\sum_{k = 0}^{\Lmax - \ell}\dfrac{\bs{r}_s^{(k)}(v)}{d(v)} - \sum_{k = 0}^{\Lmax - \ell}\dfrac{\bs{r}_t^{(k)}(v)}{d(v)}\right) & \\
    + \dfrac{1}{N}\sum_{i = 1}^{N}\sum_{\ell = 0}^{\Lmax}\sum_{v\in V}X^{(\ell)}_{t, i}(v)\left(\sum_{k = 0}^{\Lmax - \ell}\dfrac{\bs{r}_t^{(k)}(v)}{d(v)} - \sum_{k = 0}^{\Lmax - \ell}\dfrac{\bs{r}_s^{(k)}(v)}{d(v)}\right) & ,
    \end{split}
\end{align}
where $X_{s, i}^{(\ell)}(v) = \mathbb{I}\left[V_i^{(\ell)} = v \Given {V_i^{(0)}} = s\right]$ is an indicator random variable indicating whether the $\ell$-th node in the $i$-th random walk starting from $s$, $V^{(\ell)}_i$, is $v$. 
Similarly, $X_{t, i}^{(\ell)}(v)$ is defined for node $t$.

Our algorithm's core idea involves initially performing carefully designed forward push operations from nodes $s$ and $t$ until all degree-normalized residue values $\frac{\bs{r}^{(k)}(v)}{d(v)}$\footnote{Hereafter, we may occasionally omit the subscript when the statement applies to both $\bs{r}_s^{(k)}(v)$ and $\bs{r}_t^{(k)}(v)$, or both $\bs{q}_s^{(k)}(v)$ and $\bs{q}_t^{(k)}(v)$, for simplicity. If multiple subscripts are omitted in the same formula, they represent the same value.} are below a specified parameter $r_{\max}$. 
Following this, we adaptively sample a designated number of random walks from both $t$ and $s$, using the described estimator to compute an approximation of the $\Lmax$-truncated SPER.
This approach forms the basis of our Bidirectional Single-Pair Effective Resistance (BiSPER) algorithm.

The pseudo-code of our BiSPER algorithm is outlined in Algorithm \ref{alg:BiSPER}.
It operates in two phases: the \textit{push phase} and the \textit{adaptive Monte Carlo phase}. 
It invokes two procedures named BiSPER-Forward-Push and BiSPER-AMC, which  will be elaborated in subsequent sections.
\begin{algorithm}[t]
\DontPrintSemicolon
    \caption{BiSPER$(G, s, t, \Lmax, \epsilon)$}
    \label{alg:BiSPER}
    \KwIn{Graph $G = (V, E)$, node $s$ and $t$, maximum length $\Lmax$, absolute error $\epsilon$\\}
    \KwOut{$\hat{R}_{\Lmax}(s, t)$\\}
    $d \gets \min\{d(s), d(t)\}$, $r_{\max} \gets$ Equation~\eqref{eq:r-max} \\
    $\q_s^{(\ell)}, \q_t^{(\ell)} \gets \bs{0}$, $\ell = 0, 1, \cdots, \Lmax$ \\
    $\r_s^{(0)} \gets \bs{e}_s$, $\r_t^{(0)} \gets \bs{e}_t$, $\r_s^{(\ell)}, \r_t^{(\ell)} \gets \bs{0}$, $\ell = 1, 2, \cdots, \Lmax$ \\
    $Q_s[s]\text{.Update}(0, {1}/{d(s)})$, $Q_t[t]\text{.Update}(0, {1}/{d(t)})$ \\
    \For{$\ell = 0, 1, \cdots, \Lmax$}
    {
        \While{$\exists u$, such that ${\r_s^{(\ell)}(u)}\big/{d(u)} > r_{\max}$}
        {
            BiSPER-Forward-Push$_s(u, \ell)$
        }
        \While{$\exists u$, such that ${\r_t^{(\ell)}(u)}\big/{d(u)} > r_{\max}$}
        {
            BiSPER-Forward-Push$_t(u, \ell)$
        }
    }
    \tcc{Provide upper bounds for sampling random walks.}
    $T_{B_1} \gets (\Lmax + 1)(\Lmax + 2)r_{\max}$ \\
    $T_{B_2} \gets 2(\Lmax + 1) - \sum\limits_{\ell = 0}^{\Lmax}\sum\limits_{v\in V}\left(\q_s^{(\ell)}(v) + \q_t^{(\ell)}(v)\right)$ \\
    \If{$r_{\max} \ge 1 / d$}
    {
        $N \gets \left\lceil 8(\Lmax + 1)^2\log(2/p_f) / (\epsilon^2 d^2)\right\rceil$
    }
    \Else
    {
        $N \gets \left\lceil{2\min\left\{T_{B_1}^2, T_{B_2}^2\right\}\log(2 / p_f)}\big/{\epsilon^2}\right\rceil$ \\
    }
    $\hat{R}_{\Lmax} \gets \text{BiSPER-AMC}(N)$ \\
    \Return $\hat{R}_{\Lmax}$
\end{algorithm}

\subsection{Push Phase and Maintaining Prefix Sums}
In the push phase, our basic goal is to ensure all degree-normalized residues are below a threshold $r_{\max}$.
This can be easily achieved by calling Forward-Push$(u, \ell)$ for each $\ell = 0, 1, \cdots, \Lmax$ and node $u\in V$ where $\frac{\bs{r}^{(\ell)}(u)}{d(u)} > r_{\max}$.

However, according to Equation~\eqref{eq:L-truncated-estimator}, the estimator $\hat{R}_{\Lmax}(s, t)$ involves querying $\sum\limits_{k = 0}^{\Lmax - \ell}\frac{\bs{r}^{(k)}(v)}{d(v)}$, which can be seen as a prefix sum of degree-normalized residue values $\left(\frac{\bs{r}^{(k)}(v)}{d(v)}\right)_{k = 0}^{\Lmax}$.
If we query the prefix sum only when node $v$ is sampled, it introduces $O(\Lmax)$ additional cost, making each sample of random walk of length $\Lmax$ takes $O(\Lmax^2)$ time.
Pre-computing all prefix sums after the push phase costs $O(n\Lmax)$ time, which is linear to the number of nodes.
Both solutions are computationally expensive.

To address this problem, we use Binary Indexed Trees (BITs)~\cite{fenwick1994new}, also known as Fenwick Trees, to efficiently support prefix sum queries while dynamically maintaining the values in an array.
Given an array $(a_i)_{i=0}^{L - 1}$ of length $L$, a BIT is a data structure associated with it, designed to support two fundamental operations:
\begin{itemize}
    \item Update$(\ell, r)$: Increases the $\ell$-th element of the original array by $r$ and also maintain the data structure.
    \item Query$(\ell)$: Computes and returns the prefix sum up to the $\ell$-th element, i.e., $\sum\limits_{k = 0}^{\ell}a_k$.
\end{itemize}
Remarkably, both the Update$(\ell, r)$ and Query$(\ell)$ operations have a logarithmic time complexity of $O(\log L)$, where $L$ is the array length.
This logarithmic efficiency makes BITs an almost optimal choice for dynamically maintaining prefix sums in our BiSPER algorithm~\cite{pibiri2021practical}.

In the BiSPER algorithm, we employ two arrays of BITs, $Q_s$ and $Q_t$, where each entry, $Q_s[v]$ or $Q_t[v]$, represents a BIT associated with the sequences $\left(\frac{\r_s^{(\ell)}(v)}{d(v)}\right)_{\ell = 0}^{\Lmax}$ or $\left(\frac{\r_t^{(\ell)}(v)}{d(v)}\right)_{\ell = 0}^{\Lmax}$, respectively. 
Using these BITs, we design a novel forward-push operation named BiSPER-Forward-Push, detailed in Algorithm \ref{alg:BiSPER-Forward-Push}. 
This procedure redistributes residue values layer by layer, ensuring the degree-normalized values do not exceed $r_{\max}$. 
At the same time, it dynamically maintains both the residue values and their prefix sums, which are essential for the subsequent sampling process.
\begin{algorithm}[t]
\DontPrintSemicolon
    \caption{BiSPER-Forward-Push$_s(u, \ell)$}
    \label{alg:BiSPER-Forward-Push}
    $\q_s^{(\ell)}(u) \gets \q_s^{(\ell)}(u) + \r_s^{(\ell)}(u)$ \\
    \For{$v\in N(u)$}
    {
        \tcc{$Q_s[v]$ is a BIT maintaining $\sum_{k = 0}^{\ell}\frac{\bs{r}_s^{(k)}(v)}{d(v)}$.}
        $\r_s^{(\ell + 1)}(v) \gets \r_s^{(\ell + 1)}(v) + \frac{\r_s^{(\ell)}(u)}{d(u)}$ \\
        $Q_s[v]\text{.Update}(\ell + 1, \frac{\r_s^{(\ell)}(u)}{d(u)d(v)})$
    }
    $Q_s[u]\text{.Update}(\ell, -\frac{\r_s^{(\ell)}(u)}{d(u)})$ \\
    $\r_s^{(\ell)}(u) \gets 0$
\end{algorithm}

\subsection{Adaptive Monte Carlo Phase}
Once all degree-normalized residues are below the threshold $r_{\max}$, we move to the adaptive Monte Carlo phase. 
This phase involves the BiSPER-AMC procedure, detailed in Algorithm \ref{alg:BiSPER-AMC}. 
\begin{algorithm}[t]
\DontPrintSemicolon
    \caption{BiSPER-AMC$(N)$}
    \label{alg:BiSPER-AMC}
    $\hat{T}_{\text{sum}} \gets 0$, $\hat{\sigma}^2_{\text{sum}} \gets 0$ \\
    \For{$i = 1, 2, \cdots, N$}
    {
        \tcc{Sample two $\Lmax$-length random walks.}
        Generate $i$-th $\Lmax$-length random walk from $s$ and $t$: $\mathcal{W}_{s, i} = \left(V_{s, i}^{(0)} = s, V_{s, i}^{(1)}, \cdots, V_{s, i}^{(\Lmax)}\right)$, $\mathcal{W}_{t, i} = \left(V_{t, i}^{(0)} = t, V_{t, i}^{(1)}, \cdots, V_{t, i}^{(\Lmax)}\right)$ \\
        $\hat{T}_i \gets $ Equation \eqref{eq:Ti} \\
        $\hat{T}_{\text{sum}} \gets \hat{T}_{\text{sum}} + \hat{T}_i$, $\hat{\sigma}^2_{\text{sum}} \gets \hat{\sigma}^2_{\text{sum}} + \hat{T}_i^2$ \\
        $\hat{T} \gets \hat{T}_{\text{sum}} / i$, $\hat{\sigma}^2 \gets \hat{\sigma}^2_{\text{sum}} / i - \hat{T}^2$ \\ 
        \If{$\sqrt{\frac{2\hat{\sigma}^2\log(3 / p_f)}{i}} + \frac{6\min\{T_{B_1}, T_{B_2}\}\log(3 / p_f)}{i} \le \epsilon$}
        {
            \textbf{break}
        }
    }
    $\hat{R}_{\Lmax} \gets \hat{T}$ \\
    \For{$\ell = 0, 1, \cdots, \Lmax$}
    {
        $\hat{R}_{\Lmax} \gets \hat{R}_{\Lmax} + \frac{\q_s^{(\ell)}(s)}{d(s)} - \frac{\q_s^{(\ell)}(t)}{d(t)} + \frac{\q_t^{(\ell)}(t)}{d(t)} - \frac{\q_t^{(\ell)}(s)}{d(s)}$
    }
    \Return $\hat{R}_{\Lmax}$
\end{algorithm}

The BiSPER-AMC procedure samples $N$ random walks and uses the BITs, dynamically maintained during the push phase, to construct an estimator $T(s, t)$, which corresponds to the last two terms in Equation \eqref{eq:L-truncated-estimator} (formally defined in Definition \ref{def:estimator}). 

In the $i$-th iteration, we sample two random walks of length $\Lmax$ starting from $t$ and $s$. 
Let $V_{t, i}^{(\ell)}$ and $V_{s, i}^{(\ell)}$ denote the $\ell$-th sampled nodes in the $i$-th random walk. 
We construct the $i$-th sample of $T(s, t)$ as follows:
\begin{equation}
\label{eq:Ti}
    \begin{aligned}
       \hat{T}_i(s, t) 
       & = \sum_{\ell = 0}^{\Lmax}(Q_s[V_{s, i}^{(\ell)}]\text{.Query}(\Lmax - \ell) - Q_t[V_{s, i}^{(\ell)}]\text{.Query}(\Lmax - \ell) \\ 
       & + Q_t[V_{t, i}^{(\ell)}]\text{.Query}(\Lmax - \ell) - Q_s[V_{t, i}^{(\ell)}]\text{.Query}(\Lmax - \ell)). 
    \end{aligned}
\end{equation}
We then use the empirical mean of all $\hat{T}_i(s, t)$, denoted as $\hat{T}(s, t)$, as an estimator of $T(s, t)$.
Early termination is facilitated using Bernstein's inequality for empirical variance (Lemma \ref{lem:bernstein}).
After the adaptive Monte Carlo phase, $\hat{T}(s, t)$ is utilized to compute $\hat{R}_{\Lmax}(s, t)$, according to Equation \eqref{eq:L-truncated-estimator} and the definition of $T(s, t)$.

\subsection{Theoretical Analysis}
The correctness and time complexity of our BiSPER algorithm are established through the following theorems.
First, we prove that the BiSPER algorithm provides an estimator with an absolute error guarantee, as stated in Theorem \ref{thm:correctness}.
\begin{theorem}[Correctness of Approximation]
\label{thm:correctness}
    In our BiSPER algorithm, we have $$\Pr{\left|\hat{R}_{\Lmax}(s, t) - R_{\Lmax}(s, t)\right| < \epsilon} \ge 1 - p_f.$$
\end{theorem}

Theorem \ref{thm:correctness} suggests that our BiSPER algorithm produces an unbiased estimator of the $\Lmax$-truncated Effective Resistance (ER). 
For the SPER value $R(s, t)$, we can obtain an approximation with an absolute error guarantee by setting $\Lmax$ in BiSPER to the value specified in Lemma \ref{lem:L}.

Next, we bound the worst-case time complexity of the BiSPER algorithm in Theorem \ref{thm:complexity}. 
\begin{theorem}[Time Complexity]
\label{thm:complexity}
    If we set $r_{\max}$ as Equation \eqref{eq:r-max}, the time complexity of BiSPER is $\tilde{O}\left(\min\left\{ \frac{\Lmax^{7/3}}{\epsilon^{2/3}}, \frac{\Lmax^3}{\epsilon^2 d^2}, m\Lmax\right\}\right)$.
    \begin{equation}
        \label{eq:r-max}
        r_{\max} = 
        \begin{cases}
            0,\ \Lmax \ge \max\left\{\frac{m^{1/2}\epsilon d}{2\log^{1/2}(2/p_f)}, \frac{2m^{3/4}\epsilon^{1/2}}{3^{3/4}\log^{1/4}(2/p_f)}\right\}, \\
            \dfrac{1}{d},\ d \ge \max\left\{\frac{2^{5/3}(\Lmax+1)^{1/3}\log^{1/3}(2/p_f)}{3^{1/2}\epsilon^{2/3}},  \frac{2(\Lmax+1)\log^{1/2}(2/p_f)}{m^{1/2}\epsilon}\right\}, \\
            \dfrac{\epsilon^{2/3}}{2^{2/3}(\Lmax+1)^{4/3}\log^{1/3}(2/p_f)},\ \text{otherwise}.
        \end{cases}
    \end{equation}
\end{theorem}
Due to space limitations, we only state the main theorems in this section and defer the proofs to Appendices \ref{app:proof-correctness} and \ref{app:proof-complexity}.

\subsection{Discussions on Algorithm Lower Bounds and Comparison with Other Methods}
From our theoretical analysis, we conclude that the BiSPER algorithm significantly improves by $\tilde{O}\left(\frac{\Lmax^{2/3}}{\epsilon^{4/3}d^2}\right)$ over the previous best algorithms, AMC and GEER.
In this subsection, we also discuss the computational complexity lower bound of SPER algorithms.
We show that our BiSPER algorithm is near-optimal on hard instance graphs, and performs the best among all known algorithms on graphs other than hard instances.

In the full version \cite{cai2023effectivefull} of their paper \cite{cai2023effective}, Cai et al. establish a theorem that delineates the lower bound applicable to any algorithm designed for estimating the SPER:
\begin{theorem}[Theorem A.1, \cite{cai2023effectivefull}]
    There are $c_0 > 0$ and infinitely many $n$ such that given any $d \in [4, n]$ and any $\ell \in [4, n]$, for graphs of $n$ vertices and degree $d$, any local algorithm to approximate $R_G(s, t)$ with success probability $0.6$ and (relative) approximation ratio $1 + c_0\min\{d, \ell\}$ needs $\Omega(d n / \ell)$ queries.
\end{theorem}
Setting $\ell = 4$ and $d \ge \ell$, we conclude that there exists a family of graphs where all algorithms estimating SPER within an additive error of $\epsilon = 4c_0 R(s, t)$ require $\Omega(dn / \ell) = \Omega(m)$ time. 
On these hard instances, nearly-linear time Laplacian solvers~\cite{spielman2004nearly, koutis2010approaching, koutis2011nearly, cohen2014solving, jambulapati2021ultrasparse} are nearly optimal due to only poly-logarithmic overhead.
However, this worst-case near-optimality allows for improved performance on other graph families while still ensuring nearly optimal running times for hard instances.
Note that $\Lmax = O\left(\log_{\frac{1}{\lambda}} \frac{1}{\epsilon(1 - \lambda)}\right)$.
For families of graphs with a spectral radius $\lambda$ bounded by a constant less than $1$ --- commonly observed in real-world power-law graphs~\cite{qi2021real} --- we have $\Lmax = O\left(\log \frac{1}{\epsilon}\right)$ in such cases.
Under this assumption, the time complexity is $\tilde{O}\left(\frac{1}{\epsilon^2}\right)$ for the EstEff-TranProb algorithm~\cite{peng2021local}, $\tilde{O}\left(\frac{1}{\epsilon^2 d^2}\right)$ for AMC and GEER~\cite{yang2023efficient}, 
$\tilde{O}\left(\frac{m}{\epsilon^2 d}\right)$ for EstEff-MC~\cite{peng2021local}, $\tilde{O}\left(m\right)$ for Laplacian solvers~\cite{spielman2004nearly, koutis2010approaching, koutis2011nearly, cohen2014solving, jambulapati2021ultrasparse}, and $\tilde{O}\left(\min\left\{\frac{1}{\epsilon^{2/3}}, \frac{1}{\epsilon^2d^2}, m\right\}\right)$ for our BiSPER algorithm. 
Thus, under the $1 - \lambda = \Theta(1)$ assumption, BiSPER performs the best among all candidate algorithms.

It should be noted that for landmark-based algorithms, it is not feasible to set the push threshold, $r_{\max}$, and the number of samples, $T$, to guarantee that an estimator meets the desired error requirement. 
Consequently, comparing the complexity of landmark-based algorithms with our BiSPER algorithm is not straightforward. 
In Experiments II and III in Section \ref{sec:experiments}, we observe comparable empirical performance, indicating that BiSPER is at least as effective as landmark-based approaches. 
However, a critical advantage of our BiSPER algorithm is its capability to handle $\Lmax$-truncated SPER queries, a feature not supported by the landmark-based algorithms.

Although the idea of using a bidirectional method in our BiSPER algorithm is inspired by the BiMSTP algorithm in \cite{banerjee2015fast}, our BiSPER algorithm distinguishes itself from related works, including BiMSTP, in several key aspects:
\begin{itemize}
    \item \underline{Innovative Forward Push Procedure:} Our algorithm is the first to introduce a modified forward push procedure using Binary Indexed Trees to dynamically maintain the prefix sums of normalized residue values. 
    This adaptation is non-trivial and essential for designing a faster bidirectional SPER estimation algorithm.
    \item \underline{Align Bidirectional Methods with Related Works:} The BiMSTP algorithm guarantees expected running time only for randomly selected target nodes $t$. 
    Additionally, it only ensures relative error guarantees for transition probabilities above a threshold $\delta$. 
    Without positive lower bounds for MSTPs, direct application in SPER estimation with absolute error guarantees is not feasible. 
    We resolve these issues by providing a detailed analysis that establishes worst-case time complexity for any nodes $s$ and $t$. 
    This aligns our algorithm with related works while achieving more favorable complexity bounds, highlighting its theoretical and practical superiority.
\end{itemize}

\section{Experiments}
\label{sec:experiments}
In this section, we evaluate the efficiency of our BiSPER algorithm and compare it with other algorithms through experiments:
\begin{itemize}
    \item First, we compare our BiSPER algorithm with other algorithms on real-world graphs of various sizes.
    Due to large $\Lmax$ required for accurate SPER approximation, obtaining the ground truth of SPERs on large graphs (such as Friendster) using the Power Iteration algorithm is infeasible. 
    Therefore, we compare our BiSPER algorithm with other transition-probabilities-based algorithms for estimating $\Lmax$-truncated ER using a fixed maximum length of $\Lmax = 100$\footnote{Even when $\Lmax$ is set to $100$ it still takes about one week to compute the truncated SPER values on the Friendster dataset.}. This setting also ensures a fair comparison with methods designed for $\Lmax$-truncated ER queries; and it is also similar to the choices made in prior works~\cite{peng2021local, yang2023efficient}.
    \item Then we compare our BiSPER algorithm with other algorithms on the relative small graphs.
    We use Power Iteration to get the ground truths on the Facebook dataset and a synthetic Erd\H{o}s-R{\'e}nyi graph, where obtaining ground truths is feasible, and then compare the efficiency of representative transition-probabilities-based and landmark-based algorithms on them.
\end{itemize}

\subsection{Datasets and Implementation Details}
Our experiments are conducted on six real-world SNAP datasets \cite{jure2014snapnets}, varying in size as outlined in Table \ref{tab:datasets}, and a synthetic Erd\H{o}s-R{\'e}nyi graph, detailed in Experiment III.
All our experiments are conducted on a Linux server equipped with an Intel(R) Xeon(R) Silver 4114 CPU @ 2.20GHz 40-core processor and 692GB of RAM. 
Reading, loading of graphs, and allocation of space for data structures are treated as pre-processing steps and excluded from the computation of running time.
We implement all algorithms in C++ and compile them using g++ 7.5 with the \texttt{-O3} optimization flag. 
No parallelism techniques are used in our code except in the ground truth generation with Power Iteration. 
The eigenvalues $\lambda$ of all graphs are obtained via the \texttt{scipy.sparse.linalg.eigs} function, which utilizes the ARPACK~\cite{lehoucq1998arpack} algorithm.
Our code is available at: \url{https://github.com/GuanyuCui/BiSPER}.
\begin{table}[t]
    \centering
    \caption{Statistics of datasets.}
    \label{tab:datasets}
    \begin{tabular}{lrrrrrr}
        \toprule
        \textbf{Name} & $n$ & $m$ & $\dmin$ & $\dmax$ & $\dbar$ & $\lambda$ \\
        \midrule
        Facebook & 4,039 & 88,234 & 1 & 1045 & 43.69 & 0.9992 \\
        DBLP & 317,080 & 1,049,866 & 1 & 343 & 6.62 & 0.9973 \\
        Youtube & 1,134,890 & 2,987,624 & 1 & 28754 & 5.27 & 0.9980 \\
        Orkut & 3,072,441 & 117,185,083 & 1 & 33313 & 76.28 & 0.9948 \\
        LiveJournal & 3,997,962 & 34,681,189 & 1 & 14815 & 17.35 & 0.9999 \\
        Friendster & 65,608,366 & 1,806,067,135 & 1 & 5214 & 55.06 & 0.9995 \\
        \bottomrule
    \end{tabular}
\end{table}

\begin{figure*}[t]
    \centering
    \caption{Results of Experiment I. The average error of the two outliers on the Facebook dataset is less than \texttt{1e-15}.}
    \label{fig:experiment-I}
    \includegraphics[width = 0.95\linewidth]{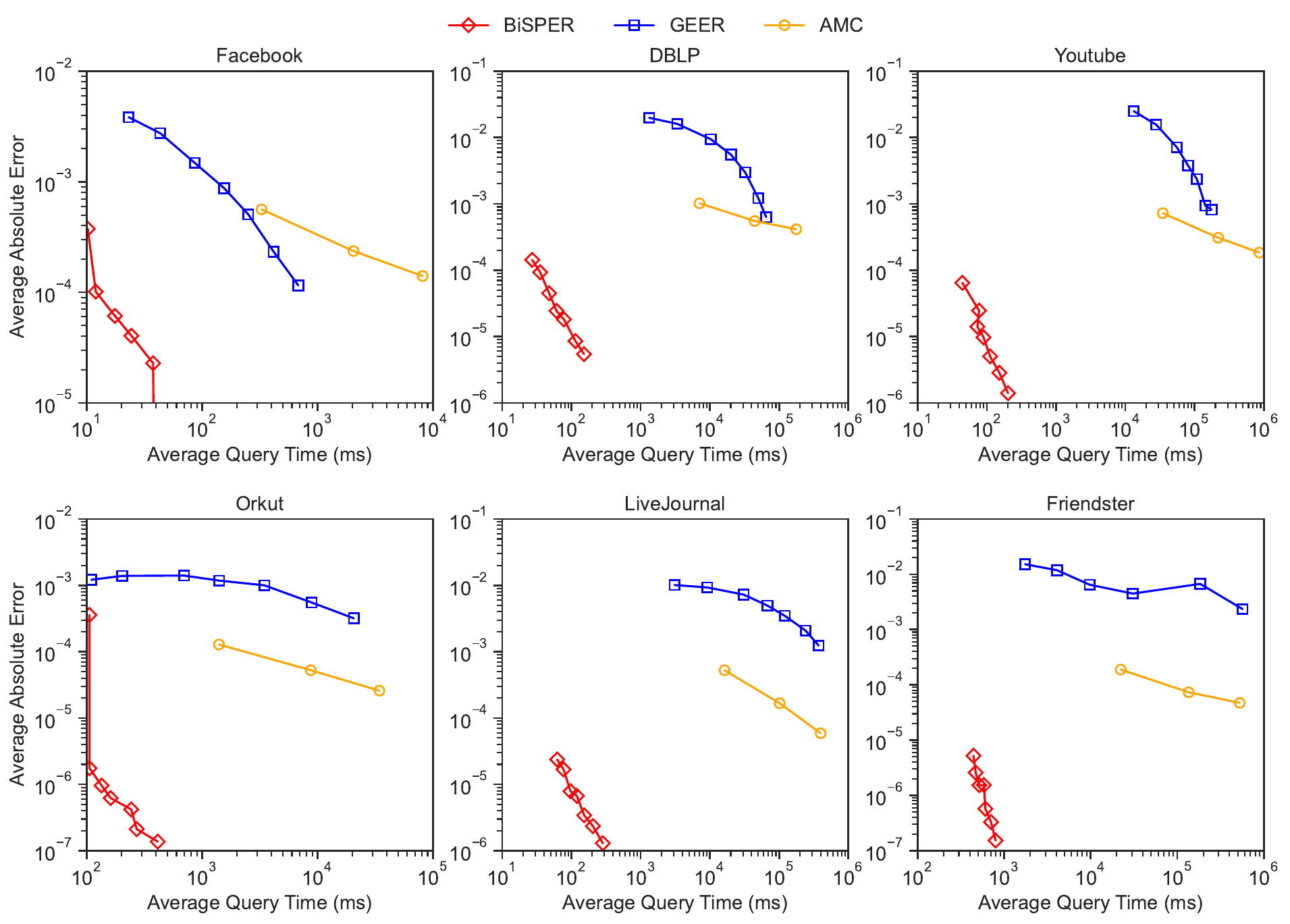}
\end{figure*}

\subsection{Experiment I: Query Efficiency for Truncated Effective Resistance on Real-World Graphs}
{\header \bf Competitors and Parameters.}
In this experiment, we aim to assess the performance of our BiSPER algorithm compared to other algorithms that use transition probabilities, specifically AMC and GEER, as introduced by Yang et al. \cite{yang2023efficient}. 
We exclude EstEff-TranProb \cite{peng2021local} as a baseline since Yang et al. \cite{yang2023efficient} have shown that AMC and GEER outperform it.

We will select 100 pairs of nodes from each graph dataset through uniform sampling, keeping these pairs constant across evaluations to ensure fair comparisons. 
The maximum random walk length, $\Lmax$, will be standardized to 100 for all algorithms.
To determine the ground truth values for ER, we will use the Power Iteration method to compute transition probabilities, followed by Equation \eqref{eq:L-truncated} to calculate the $\Lmax$-truncated ER values.

The failure probability $p_f$ is set to $0.01$ for all algorithms that allow this parameter, following the setting in \cite{yang2023efficient}. 
Other parameters for the algorithms will be varied as follows:
\begin{itemize}
    \item For BiSPER and GEER, the error parameter $\epsilon$ will be varied within $\{\texttt{1e-3}, \texttt{2e-3}, \texttt{5e-3}, \texttt{1e-2}, \texttt{2e-2}, \texttt{5e-2}, \linebreak \texttt{1e-1}\}$, except for GEER on the Friendster dataset, where it will be within $\{\texttt{1e-2}, \texttt{2e-2}, \texttt{5e-2}, \texttt{1e-1}, \texttt{2e-1}, \linebreak \texttt{5e-1}\}$ due to the rapidly increasing running time of GEER as $\epsilon$ decreases.
    \item For AMC, $\epsilon$ will be varied within $\{\texttt{1e-1}, \texttt{2e-1}, \texttt{5e-1}\}$, due to similar reason.
\end{itemize}

{\header \bf Results and Analysis.}
We run the algorithms under various parameters, recording their average running times in milliseconds and the average absolute errors across all 100 queries. 
The relationship between the average running time and the average absolute error for each dataset is visualized through line charts in Figure \ref{fig:experiment-I}.
The figure shows that BiSPER not only achieves the shortest running times, consistently below $10^3$ ms, but also maintains the lowest average absolute errors, not exceeding $10^{-3}$. 
On datasets like DBLP, Youtube, and LiveJournal, BiSPER shows a substantial speed improvement, exceeding a $10\times$ enhancement over GEER. 
Across all datasets, BiSPER outperforms AMC with more than a $10\times$ increase in speed. 
Notably, BiSPER's efficiency on graphs with lower average degrees supports our complexity analysis.

\begin{figure}[t]
    \centering
    \caption{Results of Experiment II.}
    \label{fig:experiment-II}
    \includegraphics[width = 0.5\linewidth]{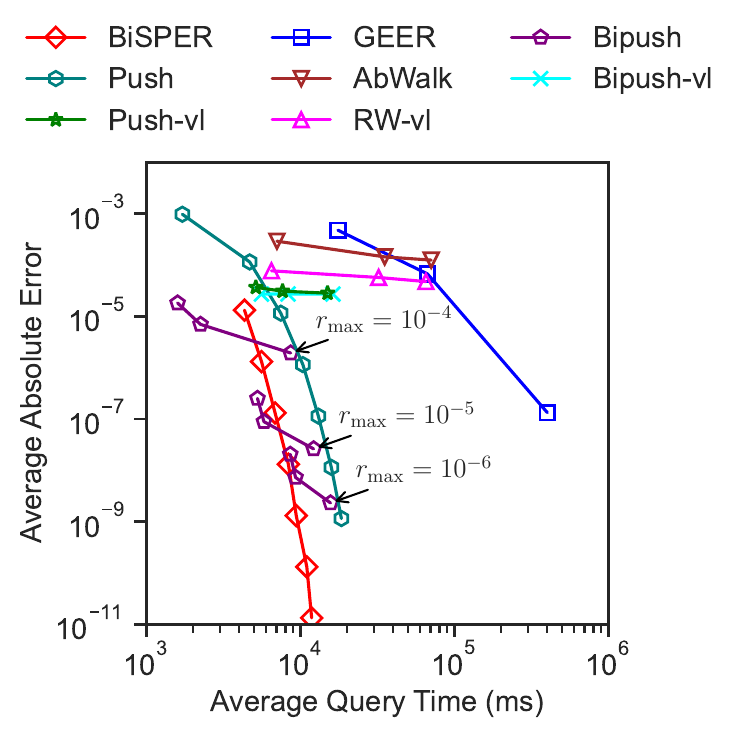}
\end{figure}

\subsection{Experiment II: Query Efficiency for Effective Resistance on Real-World Graphs}
{\header \bf Competitors and Parameters.}
In this experiment, we aim to evaluate the performance of our BiSPER algorithm against transition-probabilities-based and landmark-based methods. 
This includes the GEER algorithm by Yang et al.~\cite{yang2023efficient}, Bipush, Push and AbWalk by Liao et al.~\cite{liao2023efficient}, and Bipush-vl, Push-vl and RW-vl by Liao et al.~\cite{liao2024efficient}, using the Facebook dataset as our testbed.
AMC is excluded from this comparison due to its longer execution time, as revealed in Experiment I. 
Among the landmark-based algorithms, only Bipush, Push, and AbWalk are selected for their superior performance in~\cite{liao2023efficient}. 
We also exclude EstEff-MC~\cite{peng2021local}, following Yang et al.'s findings that AMC and GEER surpass it. 
Laplacian solvers are omitted due to the lack of practical implementations.

For consistency, we use the same 100 node pairs sampled from the Facebook graph as in the previous experiment, with the failure probability $p_f$ keeping $0.01$.
To determine the ground truth ER values, we set $\epsilon = \texttt{1e-17}$ and use Equation \eqref{eq:L} to ascertain the maximum random walk length, $L_{\max, \epsilon = \texttt{1e-17}}$. 
The $\Lmax$-truncated ER values, serving as our ground truth, are calculated using Power Iteration and Equation \eqref{eq:L-truncated}.

For the three single-landmark algorithms, the node with the largest degree is selected as the landmark. 
For the three multi-landmark algorithms, the top 100 nodes with the highest degrees are chosen as landmarks, and the number of sampled random walks is set to \texttt{1e5} to construct the index structure.
Other parameters for the algorithms are varied as follows:
\begin{itemize}
    \item For BiSPER, $\epsilon$ in $\{\texttt{1e-7}, \texttt{1e-6}, \texttt{1e-5}, \cdots, \texttt{1e-1}\}$.
    \item For GEER, $\epsilon$ in $\{\texttt{1e-3}, \texttt{1e-2}, \texttt{1e-1}\}$, due to its rapidly increasing running time as $\epsilon$ decreases.
    \item For Bipush, we vary the error parameter $r_{\max}$ and the number of samples $T$ in $\{\texttt{1e-6}, \texttt{1e-5}, \texttt{1e-4}\}$ and $\{\texttt{1e3}, \texttt{1e4}, \texttt{1e5}\}$.
    \item For Push, we vary $r_{\max}$ in $\{\texttt{1e-10}, \texttt{1e-9}, \cdots, \texttt{1e-4}\}$.
    \item For AbWalk, we vary $T$ in $\{\texttt{1e5}, \texttt{5e5}, \texttt{1e6}\}$.
    \item For Bipush-vl, we vary $r_{\max}$ in $\{\texttt{1e-7}, \texttt{2e-7}, \texttt{3e-7}\}$ and set $T$ to be 1,000.
    \item For Push-vl, we vary $r_{\max}$ in $\{\texttt{1e-7}, \texttt{2e-7}, \texttt{3e-7}\}$.
    \item For RW-vl, we vary $T$ in $\{\texttt{1e6}, \texttt{5e6}, \texttt{1e7}\}$.
\end{itemize}

{\header \bf Results and Analysis.}
As Experiment I, the relationship between the average running time and the average absolute error for each dataset is visualized through line charts in Figure \ref{fig:experiment-II}.
It reveals that BiSPER, Bipush, and Push demonstrate comparable efficiencies, with BiSPER showing superior performance, particularly when the average absolute error is below $10^{-7}$. 
This highlights BiSPER's effectiveness in achieving lower errors. 
It's also noteworthy that the landmark-based algorithms, Bipush and Push, cannot guarantee an absolute error below a pre-determined threshold $\epsilon$ through direct setting of the error parameters $r_{\max}$ or sample size $T$.

\subsection{Experiment III: Query Efficiency for Effective Resistance on Synthetic Graphs}
Real-world networks often follow power-law distributions. 
To examine how various ER algorithms perform on graphs that don't follow these distributions, we conducted a series of experiments using synthetic Erd\H{o}s-R{\'e}nyi random graphs.

{\header \bf Dataset Generation.} 
We create an Erd\H{o}s-R{\'e}nyi random graph with parameters $(n, p) = (5000, 0.005)$ and focused on its largest connected component for our synthetic dataset. 
This component consists of 5,000 nodes and 62,619 edges, featuring a minimum degree $\dmin = 8$, a maximum degree $\dmax = 44$, an average degree $\dbar = 25.05$, and $\lambda = 0.3900$.

{\header \bf Competitors and Parameters.}
Our comparison of the BiSPER algorithm uses the baseline algorithms from Experiment II as benchmarks. We varied their parameters as follows:
\begin{itemize}
    \item For BiSPER and GEER, we vary $\epsilon$ in $\{\texttt{1e-7}, \texttt{1e-6}, \cdots, \texttt{1e-1}\}$.
    \item For Bipush, we vary the error parameter $r_{\max}$ and the number of samples $T$ in $\{\texttt{1e-6}, \texttt{1e-5}, \texttt{1e-4}\}$ and $\{\texttt{1e3}, \texttt{1e4}\}$.
    \item For Push, we vary $r_{\max}$ in $\{\texttt{1e-7}, \texttt{1e-6}, \cdots, \texttt{1e-1}\}$.
    \item For AbWalk, we vary $T$ in $\{\texttt{1e3}, \texttt{5e3}, \texttt{1e4}\}$.
    \item For Bipush-vl, we vary $r_{\max}$ in $\{\texttt{1e-7}, \texttt{2e-7}, \texttt{3e-7}\}$ and set $T$ to be 1,000.
    \item For Push-vl, we vary $r_{\max}$ in $\{\texttt{1e-7}, \texttt{2e-7}, \texttt{3e-7}\}$.
    \item For RW-vl, we vary $T$ in $\{\texttt{1e4}, \texttt{5e4}, \texttt{1e5}\}$.
\end{itemize}

{\header \bf Results and Analysis.}
Consistent with earlier experiments, Figure \ref{fig:experiment-III} presents the trade-off between average absolute error and average running time across different algorithms on synthetic graphs.
As shown in the figure, our BiSPER algorithm consistently performs well. 
In contrast, landmark-based algorithms, typically strong performers on real-world graphs, struggle in regions of small absolute error. 
The lack of high-degree hub nodes in Erd\H{o}s-R{\'e}nyi graphs likely hinders effective landmark selection, leading to unguaranteed running times for these algorithms.

\begin{figure}[t]
    \centering
    \caption{Results of Experiment III.}
    \label{fig:experiment-III}
    \includegraphics[width = 0.5\linewidth]{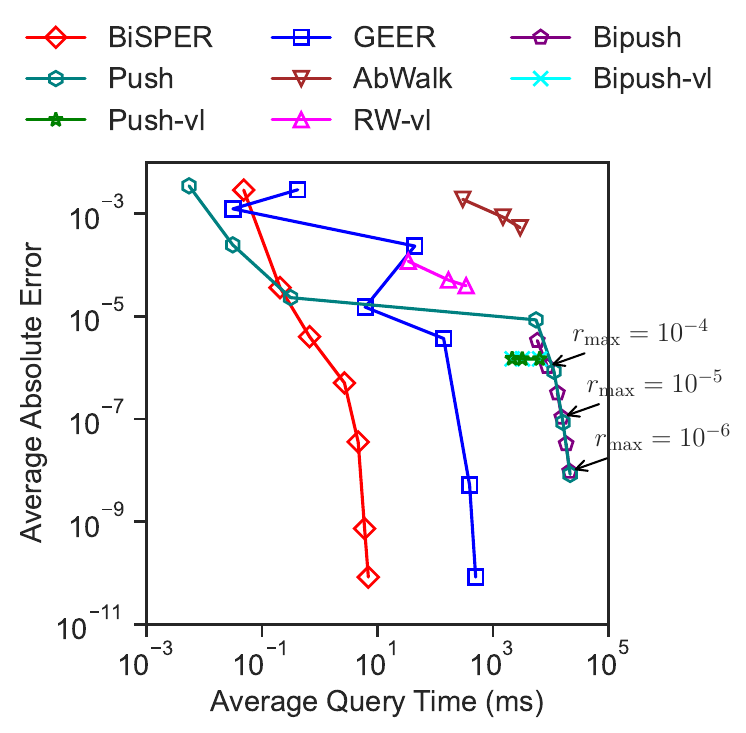}
\end{figure}
\section{Conclusion}
Effective resistance is a crucial measure of proximity in graph theory.
However, current Single-Pair Effective Resistance (SPER) estimation algorithms often suffer from high worst-case complexity, lack error guarantees, or are too theoretical and complex for practical use.
In this paper, we introduce the BiSPER algorithm, a refined bidirectional approach. 
Our theoretical analysis and extensive experimental evaluations demonstrates BiSPER's superiority over methods based on transition probabilities and commute times. 

\bibliographystyle{plain}
\bibliography{BiSPER}

\appendix
\section{Proof of Theorem \ref{thm:correctness}}
\label{app:proof-correctness}
We will first prove the unbiasedness of the estimators, and then prove the correctness of the approximation.
We begin with the definition of the estimators.
\begin{definition}[Estimator]
\label{def:estimator}
    Let $\hat{R}_{\Lmax}(s, t)$ be the $\Lmax$-truncated effective resistance estimator defined in Equation \eqref{eq:L-truncated-estimator}. 
    For convenience, we define another value $T(s, t)$ as follows:
    \begin{align*}
        T(s, t) 
        & = \sum\limits_{\ell = 0}^{\Lmax}\sum\limits_{v\in V}{p^{(\ell)}(s, v)}\left(\sum\limits_{k = 0}^{\Lmax - \ell}\frac{\bs{r}_s^{(k)}(v)}{d(v)} - \sum\limits_{k = 0}^{\Lmax - \ell}\frac{\bs{r}_t^{(k)}(v)}{d(v)}\right) \\
        & + \sum\limits_{\ell = 0}^{\Lmax}\sum\limits_{v\in V}{p^{(\ell)}(t, v)} \left(\sum\limits_{k = 0}^{\Lmax - \ell}\frac{\bs{r}_t^{(k)}(v)}{d(v)} - \sum\limits_{k = 0}^{\Lmax - \ell}\frac{\bs{r}_s^{(k)}(v)}{d(v)}\right).
    \end{align*}
    We also define each sample $\hat{T}_i(s, t)$ as follows:
    \begin{align*}
        \hat{T}_i(s, t) 
        & = \sum_{\ell = 0}^{\Lmax}\sum_{v\in V}X^{(\ell)}_{s, i}(v)\left(\sum_{k = 0}^{\Lmax - \ell}\dfrac{\bs{r}_s^{(k)}(v)}{d(v)} - \sum_{k = 0}^{\Lmax - \ell}\dfrac{\bs{r}_t^{(k)}(v)}{d(v)}\right) \\
        & + \sum_{\ell = 0}^{\Lmax}\sum_{v\in V}X^{(\ell)}_{t, i}(v)\left(\sum_{k = 0}^{\Lmax - \ell}\dfrac{\bs{r}_t^{(k)}(v)}{d(v)} - \sum_{k = 0}^{\Lmax - \ell}\dfrac{\bs{r}_s^{(k)}(v)}{d(v)}\right),
    \end{align*}
    and the estimator $\hat{T}(s, t)$, or $\hat{T}$ for short, as follows:
    \begin{align*}
        \hat{T}(s, t) 
        & = \dfrac{1}{N}\sum_{i = 1}^{N}T_i(s, t).
    \end{align*} 
\end{definition}

Next, we prove the unbiasedness of the estimators defined above.
\begin{lemma}[Unbiasedness of $\hat{T}$ and $\hat{R}_{\Lmax}$]
In Algorithm \ref{alg:BiSPER-AMC}, we have 
    $\E{\hat{T}(s, t)} = T(s, t)$ and $\E{\hat{R}_{\Lmax}(s, t)} = R_{\Lmax}(s, t)$.
\end{lemma}
\begin{proof}
    According to the definition of $\hat{T}$, we have:
    \begin{align*}
        \E{\hat{T}(s, t)} = \frac{1}{N}\sum_{i = 1}^N \E{T_i(s, t)} = T(s, t).
    \end{align*}
    And for $\hat{R}_{\Lmax}(s, t)$, we have 
    \begin{equation*}
        \begin{aligned}
            & \E{\hat{R}_{\Lmax}(s, t)} \\
            = \; & \E{\hat{T}(s, t) + \sum_{\ell = 0}^{\Lmax}\left(\dfrac{\bs{q}_s^{(\ell)}(s)}{d(s)} - \dfrac{\bs{q}_s^{(\ell)}(t)}{d(t)}\right) + \sum_{\ell = 0}^{\Lmax}\left(\dfrac{\bs{q}_t^{(\ell)}(t)}{d(t)} - \dfrac{\bs{q}_t^{(\ell)}(s)}{d(s)}\right)}.
        \end{aligned}
    \end{equation*}

    The conclusion follows if we expand the right side using the linearity of the expectation and then apply $\E{\hat{T}(s, t)} = T(s, t)$.
\end{proof}

Before we continue to the correctness of the approximation, we need the following two lemmas: Hoeffding's inequality and Bernstein's inequality for empirical variance.
\begin{lemma}[Hoeffding's Inequality, \cite{hoeffding1994probability}]
    Let $Z_1, Z_2, \cdots, Z_n$ be independent random variables bounded in the interval $[a, b]$ with length $B = b - a$ for all $1 \leq i \leq n$. 
    Then the following inequality holds:
    $$\Pr{\left|\frac{1}{n}\sum_{i = 1}^n Z_i - \frac{1}{n}\sum_{i = 1}^{n}\E{Z_i}\right| \ge \epsilon} \le 2\exp\left(-\frac{2n\epsilon^2}{B^2}\right).$$
\end{lemma}

\begin{lemma}[Bernstein's Inequality for Empirical Variance \cite{audibert2007tuning, mnih2008empirical}]
\label{lem:bernstein}
   Let $Z_1, Z_2, \cdots, Z_n$ be independent and identically distributed random variables bounded in the interval $[a, b]$ with length $B = b - a$ for all $1 \leq i \leq n$, and let $\mu = \E{Z_i}$ be their common expectation. 
   Consider the empirical expectation $\bar{Z}$ and empirical variance $\bar{\sigma}^2$ defined respectively by $\bar{Z} = \frac{1}{n}\sum\limits_{i = 1}^n Z_i$ and $\bar{\sigma}^2 = \frac{1}{n}\sum\limits_{i = 1}^n (Z_i - \bar{Z})^2$. 
   Then we have
   $$\Pr{\left|\bar{Z} - \E{Z}\right|\le \sqrt{\frac{2\bar{\sigma}^2\log(3/p_f)}{n}} + \frac{3B\log(3/p_f)}{n}} \ge 1 - p_f.$$
\end{lemma}

We also need the following bound for the sum of all residues in our proof of correctness of approximation.
\begin{lemma}[Bound of Sum of Residues]
\label{lem:sum-of-residues}
    The following bounds for the sum of residues holds in Algorithm \ref{alg:BiSPER}:
    $$\sum\limits_{v\in V}\sum\limits_{k = 0}^{\ell}\r^{(k)}(v) = 1 - \sum_{v\in V}\q^{(\ell)}(v).$$
    Hereafter, we may omit the subscript when the result applies to both $\bs{r}_s^{(k)}(v)$ and $\bs{r}_t^{(k)}(v)$, or both $\bs{q}_s^{(k)}(v)$ and $\bs{q}_t^{(k)}(v)$, for simplicity.
\end{lemma}
\begin{proof}
    According to the property of transition probability, we have $1 = \sum\limits_{v\in V}p^{(\ell)}(s, v)$.
    Substituting $p^{(\ell)}(s, v)$ with the invariant in Lemma \ref{lem:invariant}, we obtain:
    \begin{align*}
        1 & = \sum_{v\in V}\q_s^{(\ell)}(v) + \sum_{k = 0}^{\ell}\sum_{u\in V}\r_s^{(\ell - k)}(u)\sum_{v\in V}p^{(k)}(u, v) \\
        & = \sum_{v\in V}\q_s^{(\ell)}(v) + \sum_{k = 0}^{\ell}\sum_{v\in V}\r_s^{(\ell - k)}(v).
    \end{align*}
    Moving the first term to the left side of the equation yields the desired conclusion. 
    A similar conclusion holds for node $t$.
\end{proof}

With all the lemmas in place, we can proceed to prove the correctness of approximation provided by our BiSPER algorithm.
\begin{proof}
    We will abbreviate $\Lmax$ to $L$ for simplicity in our proof.

    The randomness of the estimator $\hat{R}_{L}(s, t)$ arises from the $\hat{T}(s, t)$ part.
    We first bound the number of sampled random walks ($N$ in Algorithm \ref{alg:BiSPER}) using Hoeffding's inequality.
    We will do this in the following two ways:
    \begin{enumerate}[leftmargin=*]
        \item On one hand, note that after the push phase, for any node $v$, $0\le \frac{\r^{(k)}(v)}{d(v)} \le r_{\max}$ holds. 
        Thus, we have 
        $0\le \sum\limits_{k = 0}^{L - \ell}\frac{\r^{(k)}(v)}{d(v)} \le (L + 1 - \ell)r_{\max}$.
        Therefore, if we denote $T_{B_1} = (L + 1)(L + 2)r_{\max}$, we can conclude that both the first term and second term in $\hat{T}_i$ range in $\left[-\frac{T_{B_1}}{2}, \frac{T_{B_1}}{2}\right]$, and thus $\hat{T}_i \in [-T_{B_1}, T_{B_1}]$.
        According to Hoeffding's inequality, when we set $N = \left\lceil\frac{2T_{B_1}^2\log(2 / p_f)}{\epsilon^2}\right\rceil$, we have the conclusion.

        \item On the other hand, according to Lemma \ref{lem:sum-of-residues}, we have $0\le \sum\limits_{k = 0}^{L - \ell}\sum\limits_{v\in V}\frac{\r^{(k)}(v)}{d(v)}\le \sum\limits_{k = 0}^{L - \ell}\sum\limits_{v\in V}{\r^{(k)}(v)} = 1 - \sum\limits_{v\in V}\q^{(L -\ell)}(v)$.
        Therefore, it leads to
        \begin{align*}
            \hat{T}_i(s, t) 
            & \le \sum_{\ell = 0}^{L}\sum_{k = 0}^{L - \ell}\sum_{v\in V}\left(\dfrac{\bs{r}_s^{(k)}(v)}{d(v)} + \dfrac{\bs{r}_t^{(k)}(v)}{d(v)}\right) \\
            & \le 2(L + 1) - \sum_{\ell = 0}^{L}\sum_{v\in V}\left(\q_s^{(\ell)}(v) + \q_t^{(\ell)}(v)\right)\eqqcolon T_{B_2},
        \end{align*}
        and analogously, $\hat{T}_i(s, t) \ge -T_{B_2}$, where $T_{B_2} \coloneqq 2(L + 1) - \sum\limits_{\ell = 0}^{L}\sum\limits_{v\in V}\left(\q_s^{(\ell)}(v) + \q_t^{(\ell)}(v)\right)$.
        Using Hoeffding's inequality again, when we set $N = \left\lceil \frac{2T_{B_2}^2\log(2/p_f)}{\epsilon^2}\right\rceil$, we have the conclusion.
    \end{enumerate}

    Another special case occurs when $r_{\max}\ge \frac{1}{d}$, where \\ $d = \min\{d(s), d(t)\}$. 
    In this scenario, no push operation is performed, and the only two non-zero residue values are $\bs{r}_s^{(0)}(s)$ and $\bs{r}_t^{(0)}(t)$. Consequently, we have $0 \le \sum_{k = 0}^{L - \ell} \frac{\r_s^{(k)}(v)}{d(v)} \le \frac{1}{d}$ and $0 \le \sum_{k = 0}^{L - \ell} \frac{\r_t^{(k)}(v)}{d(v)} \le \frac{1}{d}$. 
    This implies $\hat{T}_i \in \left[-\frac{2(L + 1)}{d}, \frac{2(L + 1)}{d}\right]$. 
    Once again, utilizing Hoeffding's inequality, we conclude that when $r_{\max} \ge \frac{1}{d}$, it suffices to set the number of random walks to be $N = \left\lceil\frac{8(L + 1)^2\log(2/p_f)}{\epsilon^2 d^2}\right\rceil$.

    We observe that the upper bound $N$ on the number of random walks to be sampled often significantly exceeds what is actually required. 
    To address this, we incorporate Bernstein's inequality (Lemma \ref{lem:bernstein}) in Lines 5–8 of Algorithm \ref{alg:BiSPER-AMC} for early stopping, following the approach proposed by~\cite{yang2023efficient}.  
\end{proof}

\section{Proof of Theorem \ref{thm:complexity}}
\label{app:proof-complexity}
Theorem \ref{thm:complexity} characterizes the time complexity of our BiSPER algorithm.
Before we analyze the time complexity of the BiSPER algorithm, we will prove the following lemma.
\begin{lemma}[Upper Bound of $\r_s^{(\ell)}(u)$]
\label{lem:upper-bound-r}
    For any $0\le \ell \le \Lmax$, the inequalities $\r_s^{(\ell)}(u) \le p^{(\ell)}(s, u)$ and $\r_t^{(\ell)}(u) \le p^{(\ell)}(t, u)$ hold.
\end{lemma}
\begin{proof}
    According to Lemma \ref{lem:invariant}, we have 
    \begin{align*}
        p^{(\ell)}(s, u) \ge \sum_{v\in V}\r_s^{(\ell)}(v)p^{(0)}(v, u) = \r_s^{(\ell)}(u).
    \end{align*}
    Similar results can be derived for $\r_t^{(\ell)}(u)$.
\end{proof}
Now we can present the proof of Theorem \ref{thm:complexity}.
\begin{proof}
    We will abbreviate $\Lmax$ to $L$ for simplicity in our proof.
    
    We first bound the cost of the push phase.
    On one hand, the cost of the push phase can be bounded as
    \begin{align*}
        \mathrm{Cost}_{\text{Push}}
        & = \sum_{\ell = 0}^{L}\sum_{u\in V}\mathbf{I}\left[\frac{\tilde{\bs{r}}^{(\ell)}(u)}{d(u)} > r_{\max}\right]d(u)\log(L + 1) \\
        & \le \log(L + 1)\sum_{\ell = 0}^{L}\sum_{u\in V} \frac{\tilde{\bs{r}}^{(\ell)}(u)}{r_{\max}d(u)} d(u) \\
        & \le \frac{\log(L + 1)}{r_{\max}}\sum_{\ell = 0}^{L}\sum_{u\in V}p^{(\ell)}(s, u) = \frac{(L + 1)\log(L + 1)}{r_{\max}}.
    \end{align*}
    where $\tilde{\bs{r}}^{(\ell)}(u)$ is the residue before $u$ is pushed.
    The inequality holds because $\frac{x}{a}$ is an upper bound for $\mathbf{I}\left[x \ge a\right]$, and $p^{(\ell)}(s, u)$ is an upper bound for $\tilde{\bs{r}}^{(\ell)}(u)$ (Lemma \ref{lem:upper-bound-r}).
    On the other hand, the cost can also be bounded by
    \begin{align*}
        \mathrm{Cost}_{\text{Push}}
        & = \sum_{\ell = 0}^{L}\sum_{u\in V}\mathbf{I}\left[\frac{\tilde{\bs{r}}^{(\ell)}(u)}{d(u)} \ge r_{\max}\right]d(u)\log(L + 1) \\
        & \le \log(L + 1)\sum_{\ell = 0}^{L}\sum_{u\in V}d(u) = 2m(L + 1)\log(L + 1).
    \end{align*}
    The inequality holds because $1$ is a naive upper bound for $\mathbf{I}\left[x \ge a\right]$.
    We also note that no push operation will be performed when $r_{\max} \ge \frac{1}{d} = \frac{1}{\min\{d(s), d(t)\}}$.
    Therefore, the push cost can be summarized as
    \begin{align*}
        \mathrm{Cost}_{\text{Push}}
        & = 
        \begin{cases}
            0, & r_{\max} \ge \frac{1}{d} \\
            \dfrac{(L + 1)\log(L + 1)}{r_{\max}}, & \frac{1}{2m} \le r_{\max} < \frac{1}{d} \\
            2m(L + 1)\log(L + 1), & r_{\max} < \frac{1}{2m}
        \end{cases}.
    \end{align*}
    
    Meanwhile, the cost of the adaptive Monte Carlo phase can be easily bounded by
    \begin{equation*}
        \begin{aligned}
            \mathrm{Cost}_{\text{AMC}}
            & = 
            \begin{cases}
                \dfrac{8(L + 1)^3\log(L + 1)\log(2 / p_f)}{\epsilon^2 d^2}, & r_{\max} \ge \frac{1}{d} \\
                c \coloneqq \dfrac{2(L + 1)^3(L + 2)^2r_{\max}^2\log(L + 1)\log(2 / p_f)}{\epsilon^2}, & r_{\max} < \frac{1}{d}
            \end{cases}.
        \end{aligned}
    \end{equation*}
    Considering all combinations, we get the total cost
    \begin{align*}
        \mathrm{Cost}_{\text{Total}} = 
        \begin{cases}
            \dfrac{8(L + 1)^3\log(L + 1)\log(2 / p_f)}{\epsilon^2 d^2}, & r_{\max} \ge \frac{1}{d} \\
            \dfrac{(L + 1)\log(L + 1)}{r_{\max}} + c, & \frac{1}{2m} \le r_{\max} < \frac{1}{d} \\
            2m(L + 1)\log(L + 1) + c, & r_{\max} < \frac{1}{2m}
        \end{cases}.
    \end{align*}

    Then we need to set the threshold $r_{\max}$ to minimize the total cost.
    In the first case, the total cost is irrelevant to $r_{\max}$. 
    In the second case, the total cost can be expressed as $\frac{A}{r_{\max}} + B r_{\max}^2$. 
    To minimize the upper bound of the total cost, we set $r_{\max} = (\frac{A}{2B})^{1 / 3} = \frac{\epsilon^{2/3}}{2^{2/3}(L + 1)^{2 / 3}(L + 2)^{2 / 3}\log^{1 / 3}(2 / p_f)} = \tilde{O}\left(\frac{\epsilon^{2/3}}{L^{4/3}}\right)$. 
    By doing so, the minimized upper bound of the total cost is $\frac{3}{2^{2 / 3}}\cdot A^{2/3}B^{1/3} = \tilde{O}\left(\frac{L^{7/3}}{\epsilon^{2/3}}\right)$.
    In the third case, the cost is minimized to $2m(L + 1)\log(L + 1) = \tilde{O}(mL)$ if we set $r_{\max} = 0$.
    Therefore, we have the following observations:
    \begin{enumerate}[leftmargin=*]
        \item When the maximum random walk length $L$ is large enough, or $L \ge \max\left\{\frac{m^{1/2}\epsilon d}{2\log^{1/2}(2 / p_f)}, \frac{2m^{3/4}\epsilon^{1/2}}{3^{3/4}\log^{1/4}(2 / p_f)}\right\}$, the third bound of the total cost is the minimum, we can set $r_{\max} = 0$.
        \item When $d = \min\{d(s), d(t)\}$ is large enough, or \\ $d \ge \max\left\{\frac{2^{5 / 3}(L + 1)^{1/3}\log^{1 / 3}(2 / p_f)}{3^{1/2}\epsilon^{2/3}}, \frac{2(L + 1)\log^{1 / 2}(2 / p_f)}{m^{1/2}\epsilon}\right\}$, the first bound of the total cost is the minimum, we can set $r_{\max} = \frac{1}{d}$.
        \item In other cases, we set $r_{\max} = \frac{\epsilon^{2 / 3}}{2^{2 / 3}(L + 1)^{2 / 3}(L + 2)^{2 / 3}\log^{1 / 3}(2 / p_f)} = \tilde{O}\left(\frac{\epsilon^{2 / 3}}{L^{4 / 3}}\right)$.
    \end{enumerate}
    With the above setting of $r_{\max}$, we can get the bound of the total cost $\tilde{O}\left(\min\left\{\frac{L^3}{\epsilon^2 d^2}, \frac{L^{7/3}}{\epsilon^{2/3}}, mL\right\}\right)$.
\end{proof}

\end{document}